\documentclass[a4paper,11pt]{article}
\usepackage[T1]{fontenc} 
\usepackage{float} 
\makeatletter
\gdef\@fpheader{ }
\gdef\@journal{ }
\makeatother
\RequirePackage{amsmath}
\RequirePackage{amssymb}
\RequirePackage{epsfig}
\RequirePackage{graphicx}
\RequirePackage[numbers,sort&compress]{natbib}
\RequirePackage{color}
\usepackage{cases}%
\RequirePackage[colorlinks=true
,urlcolor=blue
,anchorcolor=blue
,citecolor=blue
,filecolor=blue
,linkcolor=blue
,menucolor=blue
,pagecolor=blue
,linktocpage=true
,pdfproducer=medialab
,pdfa=true
]{hyperref}

\newif\ifnotoc\notocfalse
\newif\ifemailadd\emailaddfalse
\newif\iftoccontinuous\toccontinuousfalse
\makeatletter
\def\@subheader{\@empty}
\def\@keywords{\@empty}
\def\@abstract{\@empty}
\def\@xtum{\@empty}
\def\@dedicated{\@empty}
\def\@arxivnumber{\@empty}
\def\@collaboration{\@empty}
\def\@collaborationImg{\@empty}
\def\@proceeding{\@empty}
\def\@preprint{\@empty}

\newcommand{\subheader}[1]{\gdef\@subheader{#1}}
\newcommand{\keywords}[1]{\if!\@keywords!\gdef\@keywords{#1}\else%
\PackageWarningNoLine{\jname}{Keywords already defined.\MessageBreak Ignoring last definition.}\fi}
\renewcommand{\abstract}[1]{\gdef\@abstract{#1}}
\newcommand{\dedicated}[1]{\gdef\@dedicated{#1}}
\newcommand{\arxivnumber}[1]{\gdef\@arxivnumber{#1}}
\newcommand{\proceeding}[1]{\gdef\@proceeding{#1}}
\newcommand{\xtumfont}[1]{\textsc{#1}}
\newcommand{\correctionref}[3]{\gdef\@xtum{\xtumfont{#1} \href{#2}{#3}}}
\newcommand\jname{JHEP}
\newcommand\acknowledgments{\section*{Acknowledgments}}

\newcommand\preprint[1]{\gdef\@preprint{\hfill #1}}

\newtheorem{theorem}{Theorem}

\newenvironment{proof}[1][Proof]{\noindent\textbf{#1.} }{\ \rule{0.5em}{0.5em}}

\makeatother

\newcommand\note[2][]{%
\if!#1!%
\stepcounter{footnote}\footnotetext{#2}%
\else%
{\renewcommand\thefootnote{#1}%
\footnotetext{#2}}%
\fi}

\makeatletter
\newtoks\auth@toks
\renewcommand{\author}[2][]{%
  \if!#1!%
    \auth@toks=\expandafter{\the\auth@toks#2\ }%
  \else
    \auth@toks=\expandafter{\the\auth@toks#2$^{#1}$\ }%
  \fi
}
\makeatother
\makeatletter
\newtoks\affil@toks\newif\ifaffil\affilfalse
\newcommand{\affiliation}[2][]{%
\affiltrue
  \if!#1!%
    \affil@toks=\expandafter{\the\affil@toks{\item[]#2}}%
  \else
    \affil@toks=\expandafter{\the\affil@toks{\item[$^{#1}$]#2}}%
  \fi
}
\makeatother
\makeatletter
\newtoks\email@toks\newcounter{email@counter}%
\newcommand{\emailAdd}[1]{%
\emailaddtrue%
\ifnum\theemail@counter>0\email@toks=\expandafter{\the\email@toks, \@email{#1}}%
\else\email@toks=\expandafter{\the\email@toks\@email{#1}}%
\fi\stepcounter{email@counter}}
\newcommand{\@email}[1]{\href{mailto:#1}{\tt #1}}
\makeatother

\makeatletter
\newcommand*\collaboration[1]{\gdef\@collaboration{#1}}
\newcommand*\collaborationImg[2][]{\gdef\@collaborationImg{#2}}
\makeatletter
\newcommand\afterLogoSpace{\smallskip}
\newcommand\afterSubheaderSpace{\vskip3pt plus 2pt minus 1pt}
\newcommand\afterProceedingsSpace{\vskip21pt plus0.4fil minus15pt}
\newcommand\afterTitleSpace{\vskip23pt plus0.06fil minus13pt}
\newcommand\afterRuleSpace{\vskip23pt plus0.06fil minus13pt}
\newcommand\afterCollaborationSpace{\vskip3pt plus 2pt minus 1pt}
\newcommand\afterCollaborationImgSpace{\vskip3pt plus 2pt minus 1pt}
\newcommand\afterAuthorSpace{\vskip5pt plus4pt minus4pt}
\newcommand\afterAffiliationSpace{\vskip3pt plus3pt}
\newcommand\afterEmailSpace{\vskip16pt plus9pt minus10pt\filbreak}
\newcommand\afterXtumSpace{\par\bigskip}
\newcommand\afterAbstractSpace{\vskip16pt plus9pt minus13pt}
\newcommand\afterKeywordsSpace{\vskip16pt plus9pt minus13pt}
\newcommand\afterArxivSpace{\vskip3pt plus0.01fil minus10pt}
\newcommand\afterDedicatedSpace{\vskip0pt plus0.01fil}
\newcommand\afterTocSpace{\bigskip\medskip}
\newcommand\afterTocRuleSpace{\bigskip\bigskip}
\newlength{\affiliationsSep}
\newcommand\beforetochook{\pagestyle{myplain}\pagenumbering{roman}}

\DeclareFixedFont\trfont{OT1}{phv}{b}{sc}{11}

\renewcommand\maketitle{
\pagestyle{empty}
\thispagestyle{titlepage}
\setcounter{page}{0}
\noindent{\small\scshape\@fpheader}\@preprint\par

\afterLogoSpace
\if!\@subheader!\else\noindent{\trfont{\@subheader}}\fi
\afterSubheaderSpace
\if!\@proceeding!\else\noindent{\sc\@proceeding}\fi
\afterProceedingsSpace
{\LARGE\flushleft\sffamily\bfseries\@title\par}
\afterTitleSpace
\hrule height 1.5\p@%
\afterRuleSpace
\if!\@collaboration!\else
{\Large\bfseries\sffamily\raggedright\@collaboration}\par
\afterCollaborationSpace
\fi
\if!\@collaborationImg!\else
{\normalsize\bfseries\sffamily\raggedright\@collaborationImg}\par
\afterCollaborationImgSpace
\fi
{\bfseries\raggedright\sffamily\the\auth@toks\par}
\afterAuthorSpace
\ifaffil\begin{list}{}{%
\setlength{\leftmargin}{0.28cm}%
\setlength{\labelsep}{0pt}%
\setlength{\itemsep}{\affiliationsSep}%
\setlength{\topsep}{-\parskip}}
\itshape\small%
\the\affil@toks
\end{list}\fi
\afterAffiliationSpace
\ifemailadd 
\noindent\hspace{0.28cm}\begin{minipage}[l]{.9\textwidth}
\begin{flushleft}
\textit{E-mail:} \the\email@toks
\end{flushleft}
\end{minipage}
\else 
\PackageWarningNoLine{\jname}{E-mails are missing.\MessageBreak Plese use \protect\emailAdd\space macro to provide e-mails.}
\fi
\afterEmailSpace
\if!\@xtum!\else\noindent{\@xtum}\afterXtumSpace\fi
\if!\@abstract!\else\noindent{\renewcommand\baselinestretch{.9}\textsc{Abstract:}}\ \@abstract\afterAbstractSpace\fi
\if!\@keywords!\else\noindent{\textsc{Keywords:}} \@keywords\afterKeywordsSpace\fi
\if!\@arxivnumber!\else\noindent{\textsc{ArXiv ePrint:}} \href{http://arxiv.org/abs/\@arxivnumber}{\@arxivnumber}\afterArxivSpace\fi
\if!\@dedicated!\else\vbox{\small\it\raggedleft\@dedicated}\afterDedicatedSpace\fi
\ifnotoc\else
\iftoccontinuous\else\newpage\fi
\beforetochook\hrule
\tableofcontents
\afterTocSpace
\hrule
\afterTocRuleSpace
\fi
\setcounter{footnote}{0}
\pagestyle{myplain}\pagenumbering{arabic}
} 

\renewcommand{\baselinestretch}{1.1}\normalsize
\divide\@tempdima\baselineskip
\@tempcnta=\@tempdima
\skip\@mpfootins = \skip\footins
\renewcommand{\@dotsep}{10000}

\newcommand\ps@myplain{
\pagenumbering{arabic}
\renewcommand\@oddfoot{\hfill-- \thepage\ --\hfill}
\renewcommand\@oddhead{}}
\let\ps@plain=\ps@myplain

\newcommand\ps@titlepage{\renewcommand\@oddfoot{}\renewcommand\@oddhead{}}

\numberwithin{equation}{section}

\renewcommand\section{\@startsection{section}{1}{\z@}%
                                   {-3.5ex \@plus -1.3ex \@minus -.7ex}%
                                   {2.3ex \@plus.4ex \@minus .4ex}%
                                   {\normalfont\large\bfseries}}
\renewcommand\subsection{\@startsection{subsection}{2}{\z@}%
                                   {-2.3ex\@plus -1ex \@minus -.5ex}%
                                   {1.2ex \@plus .3ex \@minus .3ex}%
                                   {\normalfont\normalsize\bfseries}}
\renewcommand\subsubsection{\@startsection{subsubsection}{3}{\z@}%
                                   {-2.3ex\@plus -1ex \@minus -.5ex}%
                                   {1ex \@plus .2ex \@minus .2ex}%
                                   {\normalfont\normalsize\bfseries}}
\renewcommand\paragraph{\@startsection{paragraph}{4}{\z@}%
                                   {1.75ex \@plus1ex \@minus.2ex}%
                                   {-1em}%
                                   {\normalfont\normalsize\bfseries}}
\renewcommand\subparagraph{\@startsection{subparagraph}{5}{\parindent}%
                                   {1.75ex \@plus1ex \@minus .2ex}%
                                   {-1em}%
                                   {\normalfont\normalsize\bfseries}}

\def\fnum@figure{\textbf{\figurename\nobreakspace\thefigure}}
\def\fnum@table{\textbf{\tablename\nobreakspace\thetable}}

\long\def\@makecaption#1#2{%
  \vskip\abovecaptionskip
  \sbox\@tempboxa{\small #1. #2}%
  \ifdim \wd\@tempboxa >\hsize
    \small #1. #2\par
  \else
    \global \@minipagefalse
    \hb@xt@\hsize{\hfil\box\@tempboxa\hfil}%
  \fi
  \vskip\belowcaptionskip}

\renewenvironment{thebibliography}[1]{%
\begin{oldthebibliography}{#1}%
\small%
\raggedright%
\setlength{\itemsep}{5pt plus 0.2ex minus 0.05ex}%
}%
{%
\end{oldthebibliography}%
}

\begin{document}



\renewcommand{\thefootnote}{\fnsymbol{footnote}}

\title{\boldmath Calculating eigenvalues of many-body systems from partition functions}%

\author[]{Chi-Chun Zhou}
\author[*]{and Wu-Sheng Dai}\note{daiwusheng@tju.edu.cn.}


\affiliation[]{Department of Physics, Tianjin University, Tianjin 300350, P.R. China}









\abstract{A method for calculating the eigenvalue of a many-body system without solving
the eigenfunction is suggested. In many cases, we only need the knowledge of
eigenvalues rather than eigenfunctions, so we need a method solving only the
eigenvalue, leaving alone the eigenfunction. In this paper, the method is
established based on statistical mechanics. In statistical mechanics,
calculating thermodynamic quantities needs only the knowledge of eigenvalues
and then the information of eigenvalues is embodied in thermodynamic
quantities. The method suggested in the present paper is indeed a method for
extracting the eigenvalue from thermodynamic quantities. As applications, we
calculate the eigenvalues for some many-body systems. Especially, the method
is used to calculate the quantum exchange energies in quantum many-body
systems. Using the method, we also\ calculate the influence of the topological
effect on eigenvalues. Moreover, we improve the result of the relation between
the counting function and the heat kernel in literature.
}

\maketitle
\flushbottom


\section{Introduction}

Calculating the eigenvalue of a many-body system is often difficult. The most
direct way is to solve the eigenequation of the Hamiltonian $H$,
\begin{equation}
H\left\vert \phi_{n}\right\rangle =E_{n}\left\vert \phi_{n}\right\rangle .
\label{eigenequation}%
\end{equation}
Once the eigenequation is solved, both the eigenvalue $E_{n}$ and the
eigenfunction $\phi_{n}$ are solved simultaneously. Nevertheless, often one
only needs the knowledge of eigenvalues $\{E_{n}\}$ rather than
eigenfunctions. In such cases, solving eigenfunctions is redundant. This
inspires us to develop an approach which only focus on solving eigenvalues.

Statistical mechanics is essentially an averaging method. In statistical
mechanics the information of eigenfunctions is statistically averaged out;
calculating thermodynamic quantities needs only eigenvalues. For example, in
canonical ensembles, all the thermodynamic property of an $N$-particle system
is embodied in the canonical partition function which is determined only by
eigenvalues $\{E_{n}\}$:
\begin{equation}
Z\left(  \beta\right)  =\sum_{n}e^{-\beta E_{n}}, \label{partition_function}%
\end{equation}
where $\beta=\frac{1}{kT}$ with $k$ the Boltzmann constant and $T$ the temperature.

The knowledge of eigenvalues is embodied in the canonical partition function
$Z\left(  \beta\right)  $. Our problem is how to extract the eigenvalue
$E_{n}$ from the partition function (\ref{partition_function}), or, more
generally, from thermodynamic quantities.

The canonical partition function (\ref{partition_function}), from another
perspective, is the global heat kernel of the Hamiltonian operator $H$. The
global heat kernel%
\begin{equation}
K\left(  t\right)  =\sum_{n}e^{-tE_{n}} \label{kernel}%
\end{equation}
is the trace of the local heat kernel which is the Green function of the
initial-value problem of the heat-type equation \cite{vassilevich2003heat}.
Obviously, the global heat kernel is just the canonical partition function
with the replacement $t\rightarrow\beta$.

Recently, a relation between the heat kernel and the spectral counting
function is revealed \cite{dai2009number,dai2010approach}. The counting
function $N\left(  E\right)  $ counts the number of the eigenstates whose
eigenvalues are smaller than $E$. The relation between the heat kernel and the
counting function allows us to calculate the counting function $N\left(
E\right)  $ from the heat kernel $K\left(  t\right)  $, or, the canonical
partition function $Z\left(  \beta\right)  $, directly.

The eigenvalue can be calculated from the counting function $N\left(
E\right)  $ \cite{dai2009number}. This implies that one can calculate the
eigenvalue from the canonical partition function $Z\left(  \beta\right)  $ or
other thermodynamic quantities.

It is often difficult to calculate the eigenvalue of noninteracting quantum
systems and interacting classical and quantum systems. In noninteracting
quantum systems, there exist quantum exchange interactions; in classical
interacting systems, there exist classical inter-particle interactions, and in
quantum interacting systems, there exist both classical inter-particle
interactions and quantum exchange interactions \cite{dai2005hard}. The method
developed in the present paper allows us to calculated eigenvalues from the
thermodynamic quantity which is obtained by the statistical mechanical method.

In quantum many-body systems, the most important factor is the quantum
exchange interaction. The effect of quantum exchange interaction in eigenvalue
is the exchange energy. In statistical mechanics, the quantum exchange effect
is taken into account by simply employing Bose-Einstein statistics,
Fermi-Dirac statistics, and various kinds of intermediate statistics through
imposing various maximum occupation numbers: $\infty$ for Bose-Einstein
statistics, $1$ for Fermi-Dirac statistics, and an integer $n$ for Gentile
statistics
\cite{gentile1940itosservazioni,khare2005fractional,dai2012calculating,algin2017bose,dai2009intermediate,maslov2017relationship,dai2004gentile,dai2009exactly}%
. The method suggested in the present paper allows us to calculate the
exchange energy in eigenvalues from the partition function obtained in
statistical mechanics. In other words, we can calculate the exchange energy in
virtue of the statistical mechanics. For example, in quantum ideal and
interacting gases, the contribution of the exchange energy to the eigenvalue
is represented by the second virial coefficient which can be obtained by
statistical mechanical method.

In quantum mechanics, the exchange energy is reckoned in by symmetrizing or
antisymmetrizing the wavefunctions, in quantum filed theory, the exchange
energy is reckoned in by imposing the quantization condition on the fields,
while, in statistical mechanics, the exchange energy is reckoned in by only
simply setting the value of the maximum occupation number, since in
statistical mechanics only the information of eigenvalues other than the
information of wavefunctions is needed. That is to say, in statistical
mechanics, the exchange energy can be calculated relatively simply. Therefore,
calculating the exchange energy through statistical mechanics is a more simple approach.

Moreover, using the method, we calculate the influence of the topological
effect on eigenvalues for two-dimensional non-interacting classical and
quantum systems. In two-dimensional systems, the topological property is
described by connectivity which is described by the Euler-Poincar\'{e}
characteristic number \cite{kac1966can}.

In statistical mechanics, there are many solved models, in which the partition
functions are solved exactly or approximately. Using the method, we calculate
the eigenvalues for such models from the solved partition functions. In this
paper, we consider classical and quantum non-interacting systems, classical
interacting systems with the Lennard-Jones interaction and quantum interacting
systems with the hard-sphere interaction, the one-dimensional Ising models,
and the one-dimensional Potts model.

There are many studies devoted to the problem of eigenvalue spectra, such as
the eigenvalue spectrum of the Rabi model \cite{maciejewski2014full}, the
eigenvalue spectrum of the open spin-1/2 XXZ quantum chains with non-diagonal
boundary terms \cite{faldella2014complete}, the eigenvalue spectrum of the
antiperiodic spin-1/2 XXZ quantum chains \cite{niccoli2013antiperiodic}, the
statistical property of the eigenvalue spectrum
\cite{mierzejewski2013eigenvalue,lev2015absence}, the structure of the
eigenvalue spectrum \cite{callias1977spectra,christandl2014eigenvalue}, the
ground-state energy of the Heisenberg-Ising lattice \cite{yang1966ground}, the
ground state and the excited state of many-body localized Hamiltonians
\cite{yu2017finding}, and the ground state energy of a system of $N$ bosons
\cite{lewin2015bogoliubov}.

In statistical mechanics, many methods are developed for the calculation of
partition functions. For example, the canonical partition function for
quon\ statistics \cite{goodison1994canonical}, general formulas for the
canonical partition function of parastatistical systems
\cite{chaturvedi1996canonical}, the canonical partition function of the freely
jointed chain model \cite{mazars1998canonical}, the partition function of the
interacting calorons ensemble \cite{deldar2016partition}, the algorithm for
computing the exact partition function of lattice polymer models
\cite{hsieh2016efficient}, the exact partition function for the $q$-state
Potts Model \cite{chang2015exact}, the partition function for the
antiferromagnetic Ising model and the hard-core models
\cite{galanis2016inapproximability}, and the canonical partition functions for
different gaseous systems \cite{zhou2018canonical} are investigated.

The relation between the counting function and the heat kernel is the basics
of the method used in the present paper. However, the relation given by Ref.
\cite{dai2009number} neglects a special case. In this paper, we improve the
result in Ref. \cite{dai2009number}.

This paper is organized as follows. In section \ref{method}, we describe the
method of calculating the eigenvalue from the canonical partition function. In
section \ref{illustration}, we illustrate the method by examples. In sections
\ref{identicalbox} and \ref{identicalexternal}, we calculate the eigenvalue,
especially the exchange energy, of identical particles in a box and in an
external field. In section \ref{topology}, we calculate the influence of the
topological effect on eigenvalues. In sections \ref{Lennard-Jones} and
\ref{hard-sphere}, we calculate the eigenvalue of interacting particles with
the Lennard-Jones potential and the hard-sphere potential. In sections
\ref{Ising} and \ref{Potts}, we calculate the eigenvalues of the Ising system
and the Potts system. Conclusions are summarized in section \ref{conclusion}.
In the appendix, we provide a complete result for the relation between the
counting function and the heat kernel.

\section{Calculating eigenvalues from partition functions \label{method}}

In mechanics, all the dynamical informations are embodied in the Hamiltonian
$H$.

When regarding a many-body system as a mechanical system, one describes the
system by the solution of the eigenequation $H\left\vert \phi_{n}\right\rangle
=E_{n}\left\vert \phi_{n}\right\rangle $. The solution of the eigenfunction,
the eigenvalues $E_{n}$ and the eigenvectors $\left\vert \phi_{n}\right\rangle
$, contains all the informations of the Hamiltonian $H$. In fact, the
Hamiltonian $H$ can be reconstructed by the spectral representation as
$H=\sum_{n}E_{n}\left\vert \phi_{n}\right\rangle \left\langle \phi
_{n}\right\vert $.

When solving an eigenequation is difficult, we can regard a many-body system
as a thermodynamic system paying the price of losing the information of the
eigenvectors $\left\vert \phi_{n}\right\rangle $. A thermodynamic system can
be completely described by, e.g., in canonical ensembles, the partition
function $Z\left(  \beta\right)  =\operatorname*{tr}e^{-\beta H}$. In a
thermodynamic description, only the information is taken into account and the
information of the wavefunction are averaged out.

In a word, any mechanical system corresponds a thermodynamic system which
reserves only the information of eigenvalues. Although the information of the
wavefunction is lost in the thermodynamic description, the information of
eigenvalues remains. If we can extract the information of eigenvalues from the
thermodynamic quantity, we arrive at an approach solving only eigenvalues
without solving wavefunctions in the meantime. In this section, we show how to
calculate the eigenvalue of a system from the corresponding canonical
partition function in statistical mechanics.

For an eigenvalue spectrum $\left\{  E_{n}\right\}  $, the spectral counting
function $N\left(  E\right)  $ describes how many eigenstates whose
eigenvalues are smaller than $E$. In Refs.
\cite{dai2009number,dai2010approach}, a relation between the spectral counting
function $N\left(  E\right)  $ and the global heat kernel $K\left(  t\right)
$ is given: $N\left(  E\right)  =\frac{1}{2\pi i}\int_{c-i\infty}^{c+i\infty
}\frac{K\left(  t\right)  }{t}e^{tE}dt$. By the relation between the global
heat kernel $K\left(  t\right)  $ and the canonical partition function
$Z\left(  \beta\right)  $, we can calculate the counting function $N\left(
E\right)  $ from the canonical partition function $Z\left(  \beta\right)  $:%
\begin{equation}
N\left(  E\right)  =\frac{1}{2\pi i}\int_{c-i\infty}^{c+i\infty}\frac{Z\left(
\beta\right)  }{\beta}e^{\beta E}d\beta. \label{count_partition}%
\end{equation}

The $n$-th eigenvalue $E_{n}$ can be obtained from the counting function by
the equation \cite{courant2008methods}%
\begin{equation}
N\left(  E_{n}\right)  =n. \label{count}%
\end{equation}
Consequently, the eigenvalue $E_{n}$ can be solved by the following equation:%
\begin{equation}
\frac{1}{2\pi i}\int_{c-i\infty}^{c+i\infty}\frac{Z\left(  \beta\right)
}{\beta}e^{\beta E_{n}}d\beta=n.
\end{equation}

In the following, we solve eigenvalues for some many-body systems from the
corresponding canonical partition functions.

In should be emphasized that, in the relation between the counting function
and the heat kernel (canonical partition function), Eq. (\ref{count_partition}%
), there is a constant term $-\frac{1}{2}$ when $E=E_{n}$, (see Appendix
\ref{Appendix}). We will not take the contribution into account, because its
influence is often small enough to be ignored, especially for highly-excited
states, .

\section{Illustration of the method \label{illustration}}

In this section, we illustrate the method by some models whose eigenvalues are
already known, including a particle in a box, a harmonic oscillator, and
$N$\ bosonic harmonic oscillators.

\subsection{A particle in a box \label{a particle}}

A particle in a box in quantum mechanics corresponds to a classical ideal gas
confined in a box in statistical mechanics.

In the following, we illustrate the method by calculating the eigenvalue of a
particle in a box in virtue of the canonical partition function of
a\ classical ideal gas confined in a box.

In classical ideal gases, there are no classical inter-particle interactions
and quantum exchange interactions. The canonical partition function of a
classical ideal gas consisting of $N$ particles is \cite{reichl2016modern}%
\begin{equation}
Z\left(  \beta\right)  =z^{N}\left(  \beta\right)  , \label{1002}%
\end{equation}
where $z\left(  \beta\right)  $ is the single-particle partition function. The
single-particle partition function for free particles is
\cite{reichl2016modern}
\begin{equation}
z\left(  \beta\right)  =\frac{V}{\lambda^{D}}, \label{11002}%
\end{equation}
where $\lambda=h\sqrt{\frac{\beta}{2\pi m}}$ is the thermal wavelength with
$m$ the mass of the particle and $h$ the Planck constant, $V$ is the volume of
the container, and $D$ is the spatial dimension.\ The canonical partition
function of an $N$-particle classical ideal gas, by Eqs. (\ref{11002}) and
(\ref{1002}), is then
\begin{equation}
Z\left(  \beta\right)  =\left(  \frac{V}{\lambda^{D}}\right)  ^{N}.
\label{11102}%
\end{equation}

By the relation between the counting function $N\left(  \lambda\right)  $ and
the canonical partition function $Z\left(  \beta\right)  $, Eq.
(\ref{count_partition}), we can obtain the counting function,%
\begin{equation}
N\left(  E_{n}^{cl}\right)  =V^{N}\left(  \frac{2\pi m}{h^{2}}\right)
^{DN/2}\frac{1}{\Gamma\left(  1+DN/2\right)  }\left(  E_{n}^{cl}\right)
^{DN/2}, \label{1003}%
\end{equation}
where $\Gamma\left(  x\right)  $ is the Gamma function. The eigenvalue is
determined by the equation obtained by substituting Eq. (\ref{1003}) into Eq.
(\ref{count}):%
\begin{equation}
V^{N}\left(  \frac{2\pi m}{h^{2}}\right)  ^{DN/2}\frac{1}{\Gamma\left(
1+DN/2\right)  }\left(  E_{n}^{cl}\right)  ^{DN/2}=n.
\end{equation}
Solving the equation gives the eigenvalue
\begin{equation}
E_{n}^{cl}=\frac{h^{2}}{2\pi mV^{2/D}}\Gamma^{2/\left(  DN\right)  }\left(
1+\frac{DN}{2}\right)  n^{2/\left(  DN\right)  }. \label{1004}%
\end{equation}

Now let us see a familiar special case: the one-dimensional single-particle
case. In this case, $D=1$, $N=1$, and $V=L$. Eq. (\ref{1004}) then becomes%
\begin{equation}
\left.  E_{n}^{cl}\right\vert _{N=1}=\frac{h^{2}}{8m}\left(  \frac{n}%
{L}\right)  ^{2}. \label{11004}%
\end{equation}
This is just the eigenvalue of a particle in a one-dimensional periodic box
with a side length $L$. In a one-dimensional periodic box with a side length
$L$, the momentum of the particle is $p_{n}=\frac{h}{2L}n$, so the eigenvalue
(\ref{11004}) becomes%
\begin{equation}
\left.  E_{n}^{cl}\right\vert _{N=1}=\frac{p_{n}^{2}}{2m}.
\end{equation}

\subsection{One harmonic oscillator}

The harmonic oscillator in quantum mechanics corresponds to a classical ideal
harmonic oscillator gas in statistical mechanics.

In order to show the validity of the method and illustrate the method, we take
the harmonic oscillator as an example.

The eigenvalue of a harmonic oscillator is exactly known: $E_{n}=\hbar
\omega\left(  n+\frac{1}{2}\right)  $. The corresponding partition function
is
\begin{equation}
Z\left(  \beta\right)  =\sum_{n=1}^{\infty}e^{-\beta\hbar\omega\left(
n+\frac{1}{2}\right)  }=\left(  e^{\frac{1}{2}\beta\hbar\omega}-e^{-\frac
{1}{2}\beta\hbar\omega}\right)  ^{-1}. \label{A1}%
\end{equation}

Now we show how to obtain the eigenvalue from the partition function by the method.

By the relation between the counting function $N\left(  E_{n}\right)  $ and
the canonical partition function $Z\left(  \beta\right)  $, Eq.
(\ref{count_partition}), we can obtain the counting function.

First expand partition function (\ref{A1}) in power series of $e^{-\beta
\hbar\omega}$,%
\begin{equation}
Z\left(  \beta\right)  =\sum_{k=1}^{\infty}e^{-\frac{1}{2}\beta\hbar\omega
}e^{-\beta\hbar\omega k}. \label{AZho}%
\end{equation}

Substituting Eq. (\ref{AZho}) into Eq. (\ref{count_partition}) gives the
counting function:%
\begin{equation}
N\left(  E\right)  =\sum_{k=1}^{\infty}\theta\left(  E-\left(  \frac{1}%
{2}+k\right)  \hbar\omega\right)  , \label{A3a}%
\end{equation}
where $\theta\left(  x\right)  $\ is the Heaviside theta function.\textbf{\ }%
The eigenvalue is determined by $N\left(  E_{n}\right)  =n$, which directly
gives the eigenvalue of a harmonic oscillator:%
\begin{equation}
E_{n}=\hbar\omega\left(  n+\frac{1}{2}\right)  .
\end{equation}

\subsection{$N$\ bosonic harmonic oscillators}

An $N$\ bosonic harmonic oscillator system in quantum mechanics corresponds to
a bosonic harmonic oscillator gas in statistical mechanics.

In this example, we show the validity of the method.

\subsubsection{Calculating eigenvalues from partition functions}

The exact canonical partition function of a system consists of $N$ bosonic
harmonic oscillators is given by \cite{zhou2018canonical}%
\begin{equation}
Z\left(  \beta,N\right)  =\frac{1}{N!}\det\left(
\begin{array}
[c]{ccccc}%
Z\left(  \beta\right)  & -1 & 0 & ... & 0\\
Z\left(  2\beta\right)  & Z\left(  \beta\right)  & -2 & ... & 0\\
Z\left(  3\beta\right)  & Z\left(  2\beta\right)  & Z\left(  \beta\right)  &
... & ...\\
... & ... & ... & ... & -\left(  N-1\right) \\
Z\left(  N\beta\right)  & Z\left(  N\beta-\beta\right)  & Z\left(
N\beta-2\beta\right)  & ... & Z\left(  \beta\right)
\end{array}
\right)  . \label{A6}%
\end{equation}
where $Z\left(  \beta\right)  $ is the single particle partition function
given by Eq. (\ref{A1}).

As examples, consider two cases: $N=2$ and $N=3$.

\paragraph{Exact eigenvalues}

For $N=2$, the partition function given by Eq. (\ref{A6}) reads%
\begin{equation}
Z\left(  \beta,2\right)  =\frac{e^{2\hbar\omega\beta}}{\left(  e^{\hbar
\omega\beta}-1\right)  ^{2}\left(  e^{\hbar\omega\beta}+1\right)  }.
\label{A7}%
\end{equation}
Expanding $Z\left(  \beta,2\right)  $ given by Eq. (\ref{A7}) as a power
series of $e^{-\beta\hbar\omega}$ gives
\begin{equation}
Z\left(  \beta,2\right)  =\sum_{k=1}\frac{1}{4}\left[  2k+1-\left(  -1\right)
^{k}\right]  \left(  e^{-\hbar\omega\beta}\right)  ^{k}. \label{AA1}%
\end{equation}
The counting function can be obtained by substituting Eq. (\ref{AA1}) into Eq.
(\ref{A2b}):%
\begin{equation}
N\left(  E\right)  =\sum_{k=1}\frac{1}{4}\left[  2k+1-\left(  -1\right)
^{k}\right]  \theta\left(  E-k\hbar\omega\right)  , \label{AA3}%
\end{equation}

The eigenvalue can be obtained by solving Eq. (\ref{count}). It can be
directly shown that the exact solution of Eq. (\ref{count}) with the counting
function (\ref{AA3}) is%
\begin{equation}
E^{N=2}=\left(  n+1\right)  \hbar\omega\label{AA4}%
\end{equation}
with the degeneracy%
\begin{equation}
\omega\left(  E^{N=2}\right)  =\frac{1}{4}\left[  3+\left(  -1\right)
^{n}+2n\right]  . \label{AA5}%
\end{equation}
This agrees with the exact result given in Ref.
\cite{zhou2018statistical}.

For $N=3$, the partition function given by Eq. (\ref{A6}) reads%
\begin{equation}
Z\left(  \beta,3\right)  =\frac{e^{9\hbar\omega\beta/2}}{\left(
e^{\hbar\omega\beta}-1\right)  ^{3}\left(  1+e^{\hbar\omega\beta}\right)
\left(  1+e^{\hbar\omega\beta}+e^{2\hbar\omega\beta}\right)  }. \label{AA6}%
\end{equation}
Expanding $Z\left(  \beta,3\right)  $ given by Eq. (\ref{AA6}) as a
power series of $e^{-\beta\hbar\omega}$ gives\textbf{\ }%
\begin{align}
Z\left(  \beta,3\right)   &  =\sum_{n=\frac{3}{2},\frac{5}{2},\ldots}\frac
{1}{72}\left[  47+6\left(  k-\frac{3}{2}\right)  \left(  k+\frac{9}{2}\right)
-9\left(  -1\right)  ^{k-1/2}\right. \nonumber\\
&  \left.  +16\cos\left(  \frac{2\pi}{3}k-\pi\right)  \right]  \left(
e^{-\beta\hbar\omega}\right)  ^{k}. \label{AA7}%
\end{align}

The counting function can be obtained by substituting Eq. (\ref{AA7})
into Eq. (\ref{A2b}):%
\begin{align}
N\left(  E\right)   &  =\sum_{n=\frac{3}{2},\frac{5}{2},\ldots}\frac{1}%
{72}\left[  47+6\left(  k-\frac{3}{2}\right)  \left(  k+\frac{9}{2}\right)
-9\left(  -1\right)  ^{k-1/2}\right. \nonumber\\
&  \left.  +16\cos\left(  \frac{2\pi}{3}k-\pi\right)  \right]  \theta\left(
E-\hbar\omega k\right)  . \label{AA8}%
\end{align}
The eigenvalue can be obtained by solving Eq. (\ref{count}). It can be
directly shown that the exact solution of Eq. (\ref{count}) with the counting
function (\ref{AA8}) is%
\begin{equation}
E^{N=3}=\left(  n+\frac{3}{2}\right)  \hbar\omega\label{AA9}%
\end{equation}
with the degeneracy%
\begin{equation}
\omega\left(  E^{N=3}\right)  =\frac{1}{72}\left[  47+6n\left(  6+n\right)
+9\left(  -1\right)  ^{n}+16\cos\left(  \frac{2\pi}{3}n\right)  \right]  .
\label{AA10}%
\end{equation}
This agrees with the exact result given in Ref.
\cite{zhou2018statistical}.

\paragraph{Approximately smoothed eigenvalues}

Often, exactly solving the discrete eigenvalue from the equation (\ref{count})
is difficult. Instead, we can turn to seek an approximately smoothed
eigenvalues, which suffers a loss of the information of the degeneracy.

Take also the $N$\ bosonic harmonic oscillator system as an example.

For $N=2$\textbf{, }expanding the canonical partition function (\ref{AA6}) as
a power series of $\beta$ and then substituting into the counting function
(\ref{A2b}) to obtain the counting function give\
\begin{align}
N\left(  E\right)   &  =\frac{1}{2}\sum_{k\geq-2}\frac{-\left(  1+k\right)
\left(  \hbar\omega\right)  ^{k}B_{2+k}}{\left(  2+k\right)  !}\frac{1}{2\pi
i}\int_{c-i\infty}^{c+i\infty}\beta^{k-1}e^{\beta E}d\beta\nonumber\\
&  +\frac{1}{2}\sum_{k\geq-1}\frac{2^{k}\left(  \hbar\omega\right)
^{k}B_{1+k}\left(  \frac{1}{2}\right)  }{\left(  1+k\right)  !}\frac{1}{2\pi
i}\int_{c-i\infty}^{c+i\infty}\beta^{k-1}e^{\beta E}d\beta, \label{A8a}%
\end{align}
where $B_{k}$\ and $B_{k}\left(  x\right)  $\ are the Bernoulli numbers.
Using
\begin{equation}
\frac{1}{2\pi i}\int_{c-i\infty}^{c+i\infty}\beta^{k}e^{\beta E}%
d\beta=\left\{
\begin{array}
[c]{c}%
\frac{1}{\left(  -k-1\right)  !}E^{-k-1},\text{ \ \ }k<0,\\
\delta^{\left(  k\right)  }\left(  E\right)  ,\ \text{\ \ \ \ \ \ }k\geq0
\end{array}
\right.  . \label{AA3a}%
\end{equation}
we arrive at%
\begin{equation}
N\left(  E\right)  =-\frac{1}{24}+\frac{E}{4}\frac{1}{\hbar\omega}+\frac
{E^{2}}{4\hbar^{2}\omega^{2}}, \label{A1010}%
\end{equation}
where $B_{0}=1$, $B_{1}=-\frac{1}{2}$, $B_{2}=\frac{1}{6}$, $B_{0}\left(
\frac{1}{2}\right)  =1$, and $B_{1}\left(  \frac{1}{2}\right)  =0$\ are used.
Solving Eq. (\ref{count}) with\ Eq. (\ref{A1010}) gives the smoothed
eigenvalue,%
\begin{equation}
E_{n,}^{N=2}=\frac{1}{2}\left(  \sqrt{\frac{48n+5}{3}}-1\right)  \hbar\omega.
\label{A111}%
\end{equation}

Similarly, for $N=3$, the counting function is
\begin{equation}
N\left(  E\right)  =\frac{-5}{96}+\frac{13E}{144}\frac{1}{\hbar\omega}%
+\frac{E^{2}}{8\hbar^{2}\omega^{2}}+\frac{E^{3}}{36\hbar^{3}\omega^{3}}.
\label{A1111}%
\end{equation}
Solving Eq. (\ref{count}) with\ Eq. (\ref{A1111}) gives the smoothed
eigenvalue,\textbf{\ }%
\begin{align}
E_{n}^{N=3}  &  =\left\{  \frac{7}{6^{1/3}\left[  648n+\sqrt{6\left(
69984n^{2}-343\right)  }\right]  ^{1/3}}\right. \nonumber\\
&  \left.  +\frac{\left[  648n+\sqrt{6\left(  69984n^{2}-343\right)  }\right]
^{1/3}}{6^{2/3}}-\frac{3}{2}\right\}  \hbar\omega. \label{A222}%
\end{align}

The exact eigenvalues, Eqs. (\ref{AA4}), (\ref{AA5}), (\ref{AA9}), and
(\ref{AA10}), and the smoothed eigenvalue, Eq. (\ref{A1111}) and (\ref{A222})
are compared in Figures (\ref{nonka_0}) and (\ref{nonka_1}).

\ \begin{figure}[ptb]
\centering
\includegraphics[width=0.8\textwidth]{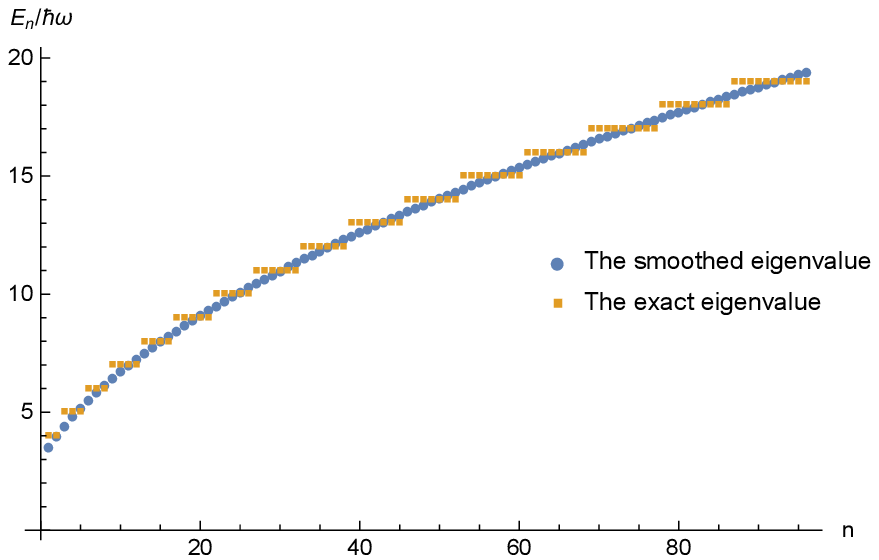}\newline\caption{Comparison of the
smoothed eigenvalue and the exact eigenvalue of a two bosonic oscillator
system. }%
\label{nonka_0}%
\end{figure}

\ \begin{figure}[ptb]
\centering
\includegraphics[width=0.8\textwidth]{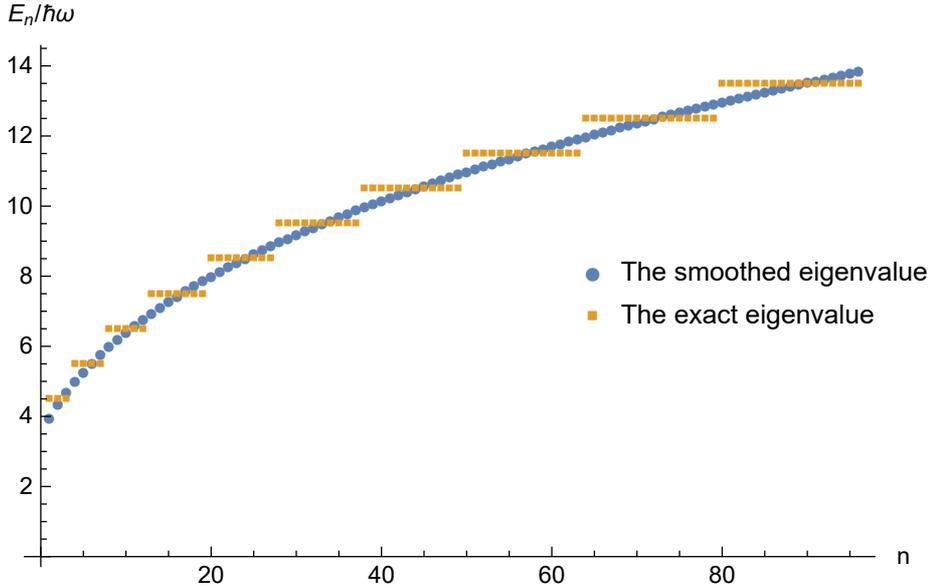}\newline\caption{Comparison of the
smoothed eigenvalue and the exact eigenvalue of a three bosonic oscillator
system. }%
\label{nonka_1}%
\end{figure}

\section{Noninteracting identical particles in a box: exchange energies
\label{identicalbox}}

A\ quantum many-body system consisting of noninteracting identical particles
corresponds to an ideal quantum gas.

There exist exchange interactions among identical particles. The influence of
the quantum exchange interaction will be reflected in the eigenvalue,
appearing as the exchange energy.

The method provides an approach to calculate the exchange energy with the help
of statistical mechanics. In statistical mechanics, the contribution of
exchange energies is taken into account by simply setting the maximum
occupation number. That is, the information of exchange interactions\ of
identical particles is embodied in the partition function of the system.

The canonical partition function of an $N$-particle quantum ideal gas is
\cite{zhou2018canonical}%
\begin{equation}
Z\left(  \beta\right)  \simeq\frac{1}{N!}z\left(  \beta\right)  ^{N}\pm
\frac{1}{2\left(  N-2\right)  !}z\left(  \beta\right)  ^{N-2}z\left(
2\beta\right)  +\ldots, \label{1008}%
\end{equation}
where "$+$" stands for Bose gases, "$-$" stands for Fermi gases, and $z\left(
\beta\right)  $ is the single-particle partition function given by Eq.
(\ref{11002}). Substituting Eq. (\ref{11002}) into Eq. (\ref{1008}) gives the
canonical partition function of a quantum ideal gas,
\begin{equation}
Z\left(  \beta\right)  \simeq\frac{1}{N!}\left(  \frac{V}{\lambda^{D}}\right)
^{N}\pm\frac{1}{2^{1+D/2}\left(  N-2\right)  !}\left(  \frac{V}{\lambda^{D}%
}\right)  ^{N-1}+\cdots. \label{1009}%
\end{equation}
The counting function is given by substituting Eq. (\ref{1009}) into Eq.
(\ref{count_partition}):%
\begin{align}
N\left(  E\right)   &  \simeq\frac{1}{N!\Gamma\left(  1+DN/2\right)  }\left(
\frac{V}{\Lambda^{D}}\right)  ^{N}E_{n}^{DN/2}\nonumber\\
&  \pm\frac{1}{2^{1+D/2}\left(  N-2\right)  !\Gamma\left(  1+D\left(
N-1\right)  /2\right)  }\left(  \frac{V}{\Lambda^{D}}\right)  ^{N-1}%
E_{n}^{DN/2-D/2}+\cdots, \label{1010}%
\end{align}
where $\Lambda=\frac{h}{\sqrt{2\pi m}}$. Then solving Eq. (\ref{count}) with
the counting function (\ref{1010}) gives\textbf{ }the eigenvalue,
\begin{align}
E_{n}  &  \simeq E_{n}^{cl}\left(  N!\right)  ^{2/\left(  DN\right)  }\left[
1+\right. \nonumber\\
&  \left.  \frac{2\left(  1-N\right)  }{D\left(  1-N\right)  ^{2}\pm
2^{1+D/2}D\left(  N!\right)  ^{1/N}\Gamma^{1/N-1}\left(  1+DN/2\right)
\Gamma\left(  1+D\left(  N-1\right)  /2\right)  n^{1/N}}\right]  ,
\label{11010}%
\end{align}
where $E_{n}^{cl}$ is the eigenvalue of a system consisting of $N$
noninteracting particles given by Eq. (\ref{1004}). We make the assumption
that the eigenvalue can be written in the form $E_{n}\sim\left(  N!\right)
^{2/\left(  DN\right)  }E_{n}^{cl}\left(  1+\text{corrections}\right)  $ with
the corrections small enough. At higher excited states, i.e., the case of
large $n$, this assumption is valid.

The influence of the quantum exchange interaction appears as the exchange
energy in eigenvalues. From the result, Eq. (\ref{11010}), we can see the
contribution of the exchange energy of bosons or fermions on the eigenvalue,
partly reflected in the terms with the sign "$\pm$".

In order to illustrate the quantum exchange interaction, we consider a
$D$-dimensional system consisting of two bosons or two fermions. By Eq.
(\ref{11010}), the eigenvalue of such a two-particle system reads%
\begin{equation}
\left.  E_{n}\right\vert _{N=2}\simeq\left.  E_{n}^{cl}\right\vert
_{N=2}\left(  2!\right)  ^{1/D}\left[  1-\frac{2\sqrt{D!}}{D\sqrt{D!}%
\pm2^{1+\left(  D+1\right)  /2}D\Gamma\left(  1+D/2\right)  n^{1/2}}\right]  ,
\label{11011}%
\end{equation}
where%
\begin{equation}
\left.  E_{n}^{cl}\right\vert _{N=2}=\frac{h^{2}}{2\pi mV^{2/D}}\Gamma
^{1/D}\left(  D+1\right)  n^{1/D}%
\end{equation}
is the eigenvalue of classical particles.

For one-dimensional cases, the eigenvalue%
\begin{equation}
\left.  E_{n}\right\vert _{N=2}^{D=1}\simeq\frac{h^{2}}{\pi mL^{2}}\left(
1-\frac{1}{1/2\pm\sqrt{\pi}n^{1/2}}\right)  n, \label{11011B}%
\end{equation}
for two-dimensional cases, the eigenvalue%
\begin{equation}
\left.  E_{n}\right\vert _{N=2}^{D=2}\simeq\frac{h^{2}}{\pi m\Omega}\left(
1-\frac{1}{1\pm4n^{1/2}}\right)  n^{1/2},
\end{equation}
and for three-dimensional cases, the eigenvalue%
\begin{equation}
\left.  E_{n}\right\vert _{N=2}^{D=3}\simeq\left(  \frac{3}{2}\right)
^{1/3}\frac{h^{2}}{\pi mV^{2/3}}\left(  1-\frac{2/3}{1\pm\sqrt{6}\sqrt{\pi
}n^{1/2}}\right)  n^{1/3},
\end{equation}
where $L$ is the length, $\Omega$ is the area, and $V$ is the volume of the container.

From Eq. (\ref{11011B}) one can see that the quantum exchange interaction
between bosons is attract and the quantum exchange interaction between
fermions is repulsive.

\section{Noninteracting identical particles in external fields
\label{identicalexternal}}

In this section, we calculate the eigenvalue of identical particles in an
external field with the help of statistical mechanics.

For an ideal gas in an external field, the single-particle eigenvalue is
determined by the Hamiltonian $H=-\frac{\hbar^{2}}{2m}\nabla^{2}+U\left(
x\right)  $. The partition function for a classical gas can be expressed as
\cite{dai2009number,bordag1996heat,vassilevich2003heat}
\begin{equation}
z\left(  \beta\right)  =\frac{V}{\lambda^{D}}\left(  1+\sum_{k=1/2,1,\ldots
}B_{k}\beta^{k}\right)  , \label{11005}%
\end{equation}
where $B_{k}$ is the heat kernel coefficient, e.g., $B_{1/2}=\frac{1}{V}\int
dxU\left(  x\right)  $. Eq. (\ref{11005}) is known as the heat kernel
expansion \cite{vassilevich2003heat}.

For quantum ideal gases in an external field, there are also quantum exchange
interactions. In order to illustrate the influence of the quantum exchange
interaction, we consider a $D$-dimensional system consisting of two bosons or
two fermions. For the two-particle case, the exact canonical partition
function is \cite{zhou2018canonical}%
\begin{equation}
Z\left(  \beta\right)  =\frac{1}{2}z^{2}\left(  \beta\right)  \pm\frac{1}%
{2}z\left(  2\beta\right)  . \label{BF2}%
\end{equation}
Substituting Eq. (\ref{11005}) into Eq. (\ref{BF2}) gives the canonical
partition function:\textbf{ }%
\begin{equation}
Z\left(  \beta\right)  \simeq\frac{1}{2}\left(  \frac{V}{\lambda^{D}}\right)
^{2}+\left(  \frac{V}{\lambda^{D}}\right)  ^{2}\sqrt{\beta}B_{1/2}\pm\frac
{1}{2^{1+D/2}}\frac{V}{\lambda^{D}}+\ldots, \label{1012}%
\end{equation}
where the second term is the contribution of the external field, the third
term is the contribution of the quantum exchange interaction.

The counting function can be obtained by substituting Eq. (\ref{1012}) into
Eq. (\ref{count_partition}):
\begin{align}
N\left(  E\right)   &  \simeq\frac{1}{2\Gamma\left(  1+D\right)  }\left(
\frac{V}{\Lambda^{D}}\right)  ^{2}E^{D}+\frac{1}{\Gamma\left(  1/2+D\right)
}\left(  \frac{V}{\Lambda^{D}}\right)  ^{2}B_{1/2}E^{D-1/2}\nonumber\\
&  \pm\frac{1}{2^{1+D/2}\Gamma\left(  1+D/2\right)  }\frac{V}{\Lambda^{D}%
}E^{D/2}+\ldots. \label{1013}%
\end{align}
Solving Eq (\ref{count}) with the counting function Eq. (\ref{1013}) gives%

\begin{equation}
\left.  E_{n}\right\vert _{N=2}\simeq E_{n}^{cl}\left(  N!\right)  ^{2/\left(
ND\right)  }\left[  1-\frac{\pm2c_{1}n^{1/\left(  2D\right)  -1/2}+2\left(
B_{1/2}\Lambda V^{-1/D}\right)  }{c_{2}n^{1/\left(  2D\right)  }\pm
Dc_{1}n^{1/\left(  2D\right)  -1/2}+\left(  2D-1\right)  \left(
B_{1/2}\Lambda V^{-1/D}\right)  }\right]  , \label{10144}%
\end{equation}
where $c_{1}=2^{-3/2+1/\left(  2D\right)  -D/2}\left(  D!\right)
^{-1/2+1/\left(  2D\right)  }\frac{\Gamma\left(  1/2+D\right)  }{\Gamma\left(
1+D/2\right)  }$ and $c_{2}=2^{1/\left(  2D\right)  }D\left(  D!\right)
^{-1+1/\left(  2D\right)  }\Gamma\left(  1/2+D\right)  $.

From Eq. (\ref{10144}), one can see the influence of the external field,
$B_{1/2}\Lambda V^{-1/D}$, and the influence of the quantum exchange
interaction, partly reflected in the terms with the sign "$\pm$", on the
eigenvalue of a quantum ideal gas.

For one-dimensional cases, the eigenvalue%
\begin{equation}
\left.  E_{n}\right\vert _{N=2}^{D=1}\simeq\frac{h^{2}}{\pi mL^{2}}\left[
1-\frac{8\left(  B_{1/2}\Lambda V^{-1/D}\right)  \pm2\sqrt{2}}{2\sqrt{2\pi
}n^{1/2}+4\left(  B_{1/2}\Lambda V^{-1/D}\right)  \pm\sqrt{2}}\right]  n,
\label{11012B}%
\end{equation}
for two-dimensional cases, the eigenvalue%
\begin{equation}
\left.  E_{n}\right\vert _{N=2}^{D=2}\simeq\frac{h^{2}}{\pi m\Omega}\left[
1-\frac{32\left(  B_{1/2}\Lambda V^{-1/D}\right)  n^{1/4}\pm3\sqrt{2\pi}%
}{12\sqrt{2\pi}n^{1/2}+48\left(  B_{1/2}\Lambda V^{-1/D}\right)  n^{1/4}%
\pm3\sqrt{2\pi}}\right]  n^{1/2},
\end{equation}
and\textbf{ }for three-dimensional cases, the eigenvalue%
\begin{align}
\left.  E_{n}\right\vert _{N=2}^{D=3}  &  \simeq\left(  \frac{3}{2}\right)
^{1/3}\frac{h^{2}}{\pi mV^{2/3}}\left[  1\right. \nonumber\\
&  \left.  -\frac{192\left(  B_{1/2}\Lambda V^{-1/D}\right)  n^{1/3}%
\pm10\times2^{5/6}\times3^{2/3}}{90\times2^{1/3}\times3^{1/6}\sqrt{\pi}%
n^{1/2}+480\left(  B_{1/2}\Lambda V^{-1/D}\right)  n^{1/3}\pm15\times
2^{5/6}\times3^{2/3}}\right]  n^{1/3}.
\end{align}

\section{Influences of topologies on eigenvalues \label{topology}}

In this section, we discuss the topology effect on the eigenvalue\textbf{ }of
classical and quantum particles in nontrivial topological containers.
Classical and quantum particles in quantum mechanics corresponds to classical
and quantum gases in statistical mechanics. In statistical mechanics, the
geometric effect and the topology effect are systematically studied
\cite{dai2003quantum,aydin2016discrete,firat2013quantum,dai2007interacting,aydin2014discrete,aydin2015dimensional}%
. In the following, we calculate the eigenvalue from the result given by
statistical mechanics.

\subsection{Non-interacting classical particles in a two-dimensional
nontrivial topological box}

The single-particle partition function of a two-dimensional ideal classical
gas in a nontrivial topological box is indeed the global heat kernel given by
Kac in his famous paper "\textit{Can one hear the shape of a drum?"
}\cite{kac1966can}. Kac's result allows us to discuss the influence of
topology of space on the eigenvalue.

The single-particle partition function of a two-dimensional confined ideal
classical gas is just the global heat kernel given by Kac \cite{kac1966can}%
\begin{equation}
z\left(  \beta\right)  =\frac{\Omega}{\lambda^{2}}-\frac{1}{4}\frac{L}%
{\lambda}+\frac{\chi}{6}, \label{110001}%
\end{equation}
where $\Omega$ is the area and $L$ the perimeter of the two-dimensional
container. $\chi=1-r$ here is the Euler-Poincar\'{e} characteristic number
with $r$ the number of holes in the two-dimensional container, which describes
the connectivity, a topological property of the system.

First consider the eigenvalue of a particle in a nontrivial topological box.
From Eq. (\ref{count_partition}), we can obtain the counting function,
\begin{equation}
N\left(  E\right)  =\frac{\Omega}{\Lambda^{2}}E-\frac{1}{2\sqrt{\pi}}\frac
{L}{\Lambda}E^{1/2}+\frac{\chi}{6}. \label{single}%
\end{equation}
Solving Eq. (\ref{count}) with\ the counting function Eq. (\ref{single}) gives
the eigenvalue%
\begin{align}
E_{n}  &  =\frac{h^{2}}{2\pi m\Omega}\left[  n+\frac{1}{2\pi}\left(  \frac
{L}{\sqrt{\Omega}}\right)  \sqrt{n\pi+\frac{1}{16}\left(  \frac{L}%
{\sqrt{\Omega}}\right)  ^{2}-\frac{\pi}{6}\chi}\right. \nonumber\\
&  \left.  +\frac{1}{8\pi}\left(  \frac{L}{\sqrt{\Omega}}\right)  ^{2}%
-\frac{1}{6}\chi\right]  . \label{lambdanKaccl}%
\end{align}
From the expression of the eigenvalue, Eq. (\ref{lambdanKaccl}), we can see
the geometric effect, reflected in the terms with the factor $\frac{L}%
{\sqrt{\Omega}}$, and the topological effect, reflected in the terms with the
factor $\chi$, explicitly.

\subsection{Non-interacting quantum particles in a two-dimensional nontrivial
topological box}

For two-particle ideal quantum systems, the canonical partition function can
be obtained by substituting Eq. (\ref{110001}) into Eq. (\ref{BF2}):%
\begin{equation}
Z\left(  \beta\right)  \simeq\frac{1}{2}\left(  \frac{\Omega}{\lambda^{2}%
}\right)  ^{2}-\frac{1}{4}\frac{L}{\lambda}\frac{\Omega}{\lambda^{2}}\pm
\frac{1}{2}\frac{\Omega}{\lambda^{2}}+\frac{1}{32}\left(  \frac{L}{\lambda
}\right)  ^{2}+\frac{\chi}{6}\frac{\Omega}{\lambda^{2}}+\ldots.
\label{1100044}%
\end{equation}

The counting function is given by substituting Eq. (\ref{1100044}) into Eq.
(\ref{count_partition}):
\begin{align}
N\left(  E\right)   &  \simeq\frac{1}{4}\left(  \frac{\Omega}{\Lambda^{2}%
}\right)  ^{2}E^{2}-\frac{1}{3\sqrt{\pi}}\frac{L}{\Lambda}\frac{\Omega
}{\Lambda^{2}}E^{3/2}\pm\frac{1}{2}\frac{\Omega}{\Lambda^{2}}E\nonumber\\
&  +\frac{1}{32}\left(  \frac{L}{\Lambda}\right)  ^{2}E+\frac{\chi}{6}%
\frac{\Omega}{\Lambda^{2}}E+\ldots. \label{110005}%
\end{align}
Solving Eq. (\ref{count}) with the counting function Eq. (\ref{110005}) gives
the eigenvalue
\begin{equation}
\left.  E_{n}\right\vert _{N=2}\simeq\frac{h^{2}}{\pi m\Omega}\left[
1-\frac{-\frac{2\sqrt{2}}{3}\frac{L}{\sqrt{\Omega}}n^{3/4}+\sqrt{\pi}\left(
\pm1+\frac{1}{3}\chi+\frac{1}{16}\frac{L^{2}}{\Omega}\right)  n^{1/2}}%
{2\sqrt{\pi}n-\sqrt{2}\frac{L}{\sqrt{\Omega}}n^{3/4}+\sqrt{\pi}\left(
\pm1+\frac{1}{3}\chi+\frac{1}{16}\frac{L^{2}}{\Omega}\right)  n^{1/2}}\right]
n^{1/2}. \label{lambdanKacclQ}%
\end{equation}

From Eq. (\ref{lambdanKacclQ}), one can see that the eigenvalue of a
two-dimensional ideal quantum gas in a nontrivial topological box is modified
by the geometric effect described by $\frac{L}{\sqrt{\Omega}}$, and by the
topological effect described by $\chi$. Moreover, for such quantum cases,
there exist exchange energies partly reflected in the terms with the sign
"$\pm$".

\section{Interacting classical many-body systems with the Lennard-Jones
interaction \label{Lennard-Jones}}

A system consisting of particles interacted with each other through the
Lennard-Jones potential\ in quantum mechanics corresponds to an interacting
gas with the Lennard-Jones interaction in statistical mechanics.

The Lennard-Jones inter-particle potential reads
\begin{equation}
U=4\varepsilon\left[  \left(  \frac{\sigma}{r}\right)  ^{12}-\left(
\frac{\sigma}{r}\right)  ^{6}\right]  , \label{LJ}%
\end{equation}
where $\varepsilon$ is the depth of the potential well and $\sigma$ is the
finite distance at which the inter-particle potential is zero
\cite{pathria2011statistical}. The canonical partition function of a classical
interacting gas with $N$ particles is given by \cite{zhou2018canonical}
\begin{equation}
Z\left(  \beta\right)  \simeq\frac{1}{N!}\left(  \frac{V}{\lambda^{3}}\right)
^{N}+\frac{1}{\left(  N-2\right)  !}\left(  \frac{V}{\lambda^{3}}\right)
^{N-1}b_{2}+\ldots, \label{1015}%
\end{equation}
where $b_{2}=\frac{2\pi}{\lambda^{3}}\int d^{3}r\left(  e^{-\beta U}-1\right)
$.

The coefficient $b_{2}$ for the Lennard-Jones interaction is
\cite{pathria2011statistical}
\begin{equation}
b_{2}\simeq\frac{2\pi}{3}\frac{r_{0}^{3}}{\lambda^{3}}\left(  u_{0}%
\beta-1\right)  , \label{111016}%
\end{equation}
where $r_{0}=2^{1/6}\sigma$. Then the canonical partition reads\textbf{ }%
\begin{equation}
Z\left(  \beta\right)  \simeq\frac{1}{N!}\left(  \frac{V}{\lambda^{3}}\right)
^{N}-\frac{2\pi}{3\left(  N-2\right)  !}\left(  \frac{V}{\lambda^{3}}\right)
^{N}\frac{r_{0}^{3}}{V}+\frac{2\pi}{3\left(  N-2\right)  !}\left(  \frac
{V}{\lambda^{3}}\right)  ^{N}\frac{r_{0}^{3}}{V}u_{0}\beta+\ldots.
\label{1017}%
\end{equation}
The counting function can be obtained by substituting Eq. (\ref{1017}) into
Eq. (\ref{count_partition}),%
\begin{align}
N\left(  E\right)   &  \simeq\frac{1}{N!\Gamma\left(  1+3N/2\right)  }\left(
\frac{V}{\Lambda^{3}}\right)  ^{N}E^{3N/2}-\frac{2\pi}{3\left(  N-2\right)
!\Gamma\left(  1+3N/2\right)  }\frac{r_{0}^{3}}{V}\left(  \frac{V}{\Lambda
^{3}}\right)  ^{N}E^{3N/2}\nonumber\\
&  +\frac{2\pi}{3\left(  N-2\right)  !\Gamma\left(  3N/2\right)  }\frac
{r_{0}^{3}}{V}u_{0}\left(  \frac{V}{\Lambda^{3}}\right)  ^{N}E^{3N/2-1}%
+\ldots. \label{1018}%
\end{align}
The eigenvalue can be obtained by solving Eq. (\ref{count}) with the counting
function (\ref{1018}):
\begin{align}
E_{n}  &  \simeq E_{n}^{cl}\left(  N!\right)  ^{2/\left(  3N\right)  }\left[
1\right. \nonumber\\
&  \left.  +\frac{4\left(  N!\right)  ^{2/\left(  3N\right)  }\Gamma
^{2/\left(  3N\right)  }\left(  1+3N/2\right)  r_{0}^{3}n^{2/\left(
3N\right)  }-6Nr_{0}^{3}V^{2/3}u_{0}\Lambda^{-2}}{3\left(  N!\right)
^{2/\left(  3N\right)  }\Gamma^{2/\left(  3N\right)  }\left(  1+3N/2\right)
\left[  \frac{3}{\left(  N-1\right)  \pi}V-2Nr_{0}^{3}\right]  n^{2/\left(
3N\right)  }+3N\left(  3N-2\right)  r_{0}^{3}V^{2/3}u_{0}\Lambda^{-2}}\right]
. \label{1019}%
\end{align}

Especially, for $N=1$, there is of course no inter-particle interactions, so
the eigenvalue (\ref{1019}) recovers the eigenvalue of a free particle%

\begin{equation}
\left.  E_{n}\right\vert _{N=1}\simeq\frac{3^{2/3}}{2^{7/3}}\frac{h^{2}}%
{m\pi^{2/3}V^{2/3}}n^{2/3}.
\end{equation}
For $N=2$, the eigenvalue (\ref{1019}) becomes%

\begin{equation}
\left.  E_{n}\right\vert _{N=2}\simeq\left(  \frac{3}{2}\right)  ^{1/3}%
\frac{h^{2}}{\pi mV^{2/3}}\left[  1+\frac{12^{1/3}\times\frac{4}{3}\pi
r_{0}^{3}n^{1/3}-4\pi r_{0}^{3}V^{2/3}u_{0}\Lambda^{-2}}{3\times
12^{1/3}\left(  V-\frac{4}{3}\pi r_{0}^{3}\right)  n^{1/3}+8\pi r_{0}%
^{3}V^{2/3}u_{0}\Lambda^{-2}}\right]  n^{1/3}.
\end{equation}

\section{Interacting quantum many-body systems with hard-sphere interactions
\label{hard-sphere}}

A system consisting of particles interacted with each other through the
hard-sphere potential\ in quantum mechanics corresponds to an interacting gas
with the hard-sphere interaction in statistical mechanics.

In a quantum hard-sphere gas, there exist both the quantum exchange
interaction and the classical hard-sphere interaction. The canonical partition
function of a quantum interacting gas is \cite{zhou2018canonical}
\begin{equation}
Z\left(  \beta\right)  \simeq\frac{1}{N!}\left(  \frac{V}{\lambda^{3}}\right)
^{N}+\frac{1}{\left(  N-2\right)  !}\left(  \frac{V}{\lambda^{3}}\right)
^{N-1}b_{2}+\ldots, \label{11016}%
\end{equation}
where
\begin{equation}
b_{2}=\frac{1}{2V\lambda^{3}}\int d^{6}qU_{2}\left(  q_{1},q_{2}\right)
\label{11017}%
\end{equation}
with $U_{2}\left(  q_{1},q_{2}\right)  =\lambda^{6}\left[  \langle q_{1}%
,q_{2}|e^{-\beta H_{2}}|q_{1},q_{2}\rangle-\langle q_{1}|e^{-\beta H_{1}%
}|q_{1}\rangle\langle q_{2}|e^{-\beta H_{1}}|q_{2}\rangle\right]  $,
$H_{1}=-\frac{\hbar^{2}}{2m}\nabla^{2}$, $H_{2}=-\frac{\hbar^{2}}{2m}%
\nabla_{1}^{2}-\frac{\hbar^{2}}{2m}\nabla_{2}^{2}+4\pi a\hbar^{2}%
/m\delta\left(  q_{1}-q_{2}\right)  $, and $a$ is the radius of the particle.

For Bose gases, Eq. (\ref{11017}) becomes $b_{2}=2^{-5/2}-2\frac{a}{\lambda
}-\frac{10\pi^{2}}{3}\frac{a^{5}}{\lambda^{5}}$ \cite{pathria2011statistical};
for Fermi gases, Eq. (\ref{11017}) becomes $b_{2}=-2^{-5/2}-6\pi\frac{a^{3}%
}{\lambda^{3}}+18\pi^{2}\frac{a^{5}}{\lambda^{5}}$
\cite{pathria2011statistical}.

Then the canonical partition function of a Bose gas is
\begin{equation}
Z_{B}\left(  \beta\right)  \simeq\frac{1}{N!}\left(  \frac{V}{\lambda^{3}%
}\right)  ^{N}+\frac{1}{\left(  N-2\right)  !}\left(  \frac{V}{\lambda^{3}%
}\right)  ^{N-1}\left(  2^{-5/2}-2\frac{a}{\lambda}-\frac{10\pi^{2}}{3}%
\frac{a^{5}}{\lambda^{5}}\right)  +\ldots, \label{1022}%
\end{equation}
and the canonical partition function of a Fermi gas is
\begin{equation}
Z_{F}\left(  \beta\right)  \simeq\frac{1}{N!}\left(  \frac{V}{\lambda^{3}%
}\right)  ^{N}+\frac{1}{\left(  N-2\right)  !}\left(  \frac{V}{\lambda^{3}%
}\right)  ^{N-1}\left(  -2^{-5/2}-6\pi\frac{a^{3}}{\lambda^{3}}+18\pi^{2}%
\frac{a^{5}}{\lambda^{5}}\right)  +\ldots, \label{1023}%
\end{equation}
respectively.

Then the counting function of a Bose gas by Eq. (\ref{count_partition}) is%
\begin{align}
&  N_{B}\left(  E\right) \nonumber\\
&  \simeq\frac{\left(  2\pi\right)  ^{3N/2}}{N!\Gamma\left(  1+3N/2\right)
}\frac{m^{3N/2}V^{N}}{h^{3N}}E^{3N/2}+\frac{2^{3N/2-4}\pi^{3\left(
N-1\right)  /2}}{\left(  N-2\right)  !\Gamma\left(  3N/2-1/2\right)  }%
\frac{m^{3\left(  N-1\right)  /2}V^{N-1}}{h^{3\left(  N-1\right)  }%
}E^{3N/2-3/2}\\
&  -\frac{2^{3N/2}\pi^{\left(  3N-2\right)  /2}a}{\left(  N-2\right)
!\Gamma\left(  3N/2\right)  }\frac{m^{\left(  3N-2\right)  /2}V^{N-1}%
}{h^{3N-2}}E^{3N/2-1}-\frac{10a^{5}\left(  2\pi\right)  ^{3N/2+3}}{3\left(
N-2\right)  !\Gamma\left(  3N/2+2\right)  }\frac{m^{\left(  3N+2\right)
/2}V^{N-1}}{h^{3N+2}}E^{3N/2+1}+\ldots, \label{1024}%
\end{align}
and the counting function of a Fermi gas by Eq. (\ref{count_partition}) is%
\begin{align}
&  N_{F}\left(  E\right) \nonumber\\
&  \simeq\frac{\left(  2\pi\right)  ^{3N/2}}{N!\Gamma\left(  1+3N/2\right)
}\frac{m^{3N/2}V^{N}}{h^{3N}}E^{3N/2}-\frac{2^{3N/2-4}\pi^{3N/2-3/2}}{\left(
N-2\right)  !\Gamma\left(  3N/2-1/2\right)  }\frac{m^{3N/2-\frac{3}{2}V^{N-1}%
}}{h^{3N-3}}E^{3N/2-3/2}\\
&  -\frac{3\times2^{3N/2+1}\pi^{3N/2+1}a^{3}}{\left(  N-2\right)
!\Gamma\left(  3N/2+1\right)  }\frac{m^{3N/2}V^{N-1}}{h^{3N}}E^{3N/2}%
+\frac{9\times2^{3N/2+2}\pi^{3N/2+3}a^{5}}{\left(  N-2\right)  !\Gamma\left(
3N/2+2\right)  }\frac{m^{3N/2+1}V^{N-1}}{h^{3N+2}}E^{3N/2+1}+\ldots.
\label{1025}%
\end{align}
Solving Eq. (\ref{count}) gives the eigenvalue of a hard-sphere Bose gas,%
\begin{equation}
E_{n}\simeq E_{n}^{cl}\left(  N!\right)  ^{2/3N}\left[  1+\frac{h_{1}%
\frac{a^{5}}{V^{5/3}}n^{5/\left(  3N\right)  }+2h_{2}\frac{a}{V^{1/3}%
}n^{1/\left(  3N\right)  }-2h_{3}}{h_{1}\frac{a^{5}}{V^{5/3}}n^{5/\left(
3N\right)  }+h_{4}n^{1/N}+5h_{2}\frac{a}{V^{1/3}}n^{1/\left(  3N\right)
}+3\left(  N-1\right)  h_{3}}\right]  \label{1026}%
\end{equation}
with $h_{1}=\frac{10\pi^{2}\left(  N!\right)  ^{5/\left(  3N\right)  }%
\Gamma^{5/\left(  3N\right)  }\left(  1+3N/2\right)  }{3\Gamma\left(
3N/2+2\right)  }$, $h_{2}=\frac{2\left(  N!\right)  ^{1/\left(  3N\right)
}\Gamma^{1/\left(  3N\right)  }\left(  1+3N/2\right)  }{\Gamma\left(
3N/2\right)  }$, $h_{3}=\frac{1}{2^{5/2}\Gamma\left(  3N/2-1/2\right)  }$, and
$h_{4}=\frac{3\left(  N!\right)  ^{1/N}\Gamma^{1/N-1}\left(  1+3N/2\right)
}{N-1}$. The eigenvalue of a hard-sphere Fermi gas,%
\begin{equation}
E_{n}\simeq E_{n}^{cl}\left(  N!\right)  ^{2/\left(  3N\right)  }\left[
1+\frac{-2t_{1}\frac{a^{5}}{V^{5/3}}n^{5/\left(  3N\right)  }+2t_{2}%
\frac{a^{3}}{V}n^{1/N}+2t_{3}}{\left(  2+3N\right)  t_{1}\frac{a^{5}}{V^{5/3}%
}n^{5/\left(  3N\right)  }+t_{4}n^{1/N}-3Nt_{2}\frac{a^{3}}{V}n^{1/N}-3\left(
N-1\right)  t_{3}}\right]  \label{1027}%
\end{equation}
with $t_{1}=\frac{18\pi^{2}\left(  N!\right)  ^{5/\left(  3N\right)  }%
\Gamma^{5/\left(  3N\right)  }\left(  1+3N/2\right)  }{\Gamma\left(
3N/2+2\right)  }$, $t_{2}=\frac{6\pi\left(  N!\right)  ^{1/\left(  3N\right)
}\Gamma^{1/\left(  3N\right)  }\left(  1+3N/2\right)  }{\Gamma\left(
3N/2+1\right)  }$, $t_{3}=\frac{1}{2^{5/2}\Gamma\left(  3N/2-1/2\right)  }$,
and $t_{4}=3\frac{\left(  N!\right)  ^{1/N}\Gamma^{1/N-1}\left(
1+3N/2\right)  }{N-1}$.

In order to illustrate the influence of the quantum exchange interaction, we
consider a system consisting of two bosons or two fermions. For the Bose case,
the eigenvalue (\ref{11011}) becomes%
\begin{align}
&  \left.  E_{n}\right\vert _{N=2}^{Bose}\overset{D=3}{\sim}12^{1/3}%
\frac{h^{2}}{2\pi mV^{2/3}}n^{1/3}\left[  1\right. \nonumber\\
&  \left.  +\frac{10\times2^{1/3}\times3^{5/6}\pi^{5/2}\frac{a^{5}}{V^{5/3}%
}n^{5/6}+36\times3^{1/6}\sqrt{\pi}\frac{a}{V^{1/3}}n^{1/6}-6\times2^{1/6}%
}{-40\times2^{1/3}\times3^{5/6}\pi^{5/2}\frac{a^{5}}{V^{5/3}}n^{5/6}%
-72\times3^{1/6}\sqrt{\pi}\frac{a}{V^{1/3}}n^{1/6}+9\times2^{1/6}\sqrt{6\pi
}\sqrt{n}+9\times2^{1/6}}\right]  ;
\end{align}
for the Fermi case, the eigenvalue (\ref{11011}) becomes%
\begin{align}
&  \left.  E_{n}\right\vert _{N=2}^{Fermi}\overset{D=3}{\sim}12^{1/3}%
\frac{h^{2}}{2\pi mV^{2/3}}n^{1/3}\left[  1\right. \\
&  \left.  +\frac{-18\times2^{2/3}\times3^{5/6}\pi^{5/2}\frac{a^{5}}{V^{5/3}%
}n^{5/6}+24\sqrt{3}\pi^{3/2}\frac{a^{3}}{V}\sqrt{n}+2\sqrt{2}}{72\times
2^{2/3}\times3^{5/6}\pi^{5/2}\frac{a^{5}}{V^{5/3}}n^{5/6}-72\sqrt{3}\pi
^{3/2}\times\frac{a^{3}}{V}\sqrt{n}+6\sqrt{3\pi}\sqrt{n}-3\sqrt{2}}\right]  .
\end{align}

\section{The one-dimensional Ising model \label{Ising}}

The eigenvalue of the one-dimensional Ising model without interactions can be
calculated from the corresponding canonical partition function directly.

\subsection{The Ising model without interactions}

For a one-dimensional $N$-particle Ising model without the interaction between
spins, the canonical partition function reads \cite{reichl2016modern}%
\begin{equation}
Z(\beta)=(e^{\beta B\mu}+e^{-B\mu\beta})^{N}, \label{1028}%
\end{equation}
where $B$ is the magnetic field strength and $\mu$ is the spin magnetic moment.

The counting function can be obtained by substituting Eq. (\ref{1028}) into
Eq. (\ref{count_partition}),%
\begin{equation}
N(E)=\sum_{s=1}^{N+1}\binom{N}{s-1}\theta\left(  E-B\mu\left[  2\left(
s-1\right)  -N\right]  \right)  , \label{1029}%
\end{equation}
where $\theta\left(  x\right)  $ is the step function. Substituting Eq.
(\ref{1029}) into Eq. (\ref{count}) and solving the equation give the
eigenvalue
\begin{equation}
E=B\mu\left[  2\left(  n-1\right)  -N\right]  \label{1030}%
\end{equation}
and the degree of degeneracy%
\begin{equation}
\omega(E)=\binom{N}{n-1}. \label{1031}%
\end{equation}

\subsection{The Ising model with nearest-neighbour interactions}

For a one-dimensional $N$-particle Ising model with nearest-neighbour
interactions, the Hamiltonian is $H=-\sum_{i}B\mu\sigma_{i}+\sum_{i}%
\epsilon\sigma_{i}\sigma_{i+1}$, where $\epsilon$ is the coupling constant.
The canonical partition function reads \cite{feynman1998statistical}%
\begin{equation}
Z\left(  \beta,N\right)  =\lambda_{+}^{N}+\lambda_{-}^{N}, \label{1033}%
\end{equation}
where%
\begin{equation}
\lambda_{\pm}=e^{\beta\epsilon}\left[  \cosh\left(  \beta\mu B\right)
\pm\sqrt{\cosh^{2}\left(  \beta\mu B\right)  -2e^{-2\beta\epsilon}\sinh\left(
2\beta\epsilon\right)  }\right]  .
\end{equation}

The counting function can be obtained by substituting Eq. (\ref{1033}) into
Eq. (\ref{count_partition}). In the following, we list the eigenvalues for
$N=3,4,5,\ldots,9$.

For example, for $N=3$, the counting function is%
\begin{align}
N\left(  E\right)   &  =\theta\left(  E-\left(  -3B\mu-3\epsilon\right)
\right)  +3\theta\left(  E-\left(  -B\mu+\epsilon\right)  \right) \nonumber\\
&  +3\theta\left(  E-\left(  B\mu+\epsilon\right)  \right)  +\theta\left(
E-\left(  3B\mu-3\epsilon\right)  \right)  ;
\end{align}
the eigenvalue $E$ and the degree of degeneracy $\omega$ are listed in Table
\ref{table1}.

\begin{table}[ptb]
\caption{The eigenvalue $E$ and the degree of degeneracy $\omega$ : $N=3$}%
\label{table1}%
\centering
\begin{tabular}
[c]{ccccc}\hline
$E$ & $-3B\mu-3\epsilon$ & $-B\mu+\epsilon$ & $B\mu+\epsilon$ & $3B\mu
-3\epsilon$\\\hline
$\omega$ & $1$ & $3$ & $3$ & $1$\\\hline
\end{tabular}
\end{table}

For other values of $N$, we only list the eigenvalues in Tables \ref{table2}%
-\ref{table7}.

\begin{table}[ptb]
\caption{The eigenvalue $E$ and the degree of degeneracy $\omega$ : $N=4$}%
\label{table2}%
\centering
\begin{tabular}
[c]{ccccccc}\hline
$E$ & $-4B\mu-4\epsilon$ & $-2B\mu$ & $0$ & $4\epsilon$ & $2B\mu$ &
$4B\mu-4\epsilon$\\\hline
$\omega$ & $1$ & $4$ & $4$ & $2$ & $4$ & $1$\\\hline
\end{tabular}
\end{table}\begin{table}[ptb]
\caption{The eigenvalue $E$ and the degree of degeneracy $\omega$ : $N=5$}%
\label{table3}%
\centering
\begin{tabular}
[c]{ccccccccc}\hline
$E$ & $-5B\mu-5\epsilon$ & $-3B\mu-\epsilon$ & $-B\mu+3\epsilon$ &
$-B\mu-\epsilon$ & $B\mu+3\epsilon$ &  &  & \\\hline
$\omega$ & $1$ & $5$ & $5$ & $5$ & $5$ &  &  & \\
$E$ & $B\mu-\epsilon$ & $3B\mu-\epsilon$ & $5B\mu-5\epsilon$ &  &  &  &  &
\\\hline
$\omega$ & $5$ & $5$ & $1$ &  &  &  &  & \\\hline
\end{tabular}
\end{table}

\begin{table}[ptb]
\caption{The eigenvalue $E$ and the degree of degeneracy $\omega$ : $N=6$}%
\label{table4}%
\centering
\begin{tabular}
[c]{cccccccccccc}\hline
$E$ & $-6B\mu-6\epsilon$ & $-4B\mu-2\epsilon$ & $-2B\mu+2\epsilon$ &
$-2B\mu-2\epsilon$ & $-2\epsilon$ &  &  &  &  &  & \\\hline
$\omega$ & $1$ & $6$ & $9$ & $6$ & $6$ &  &  &  &  &  & \\
$E$ & $2\epsilon$ & $6\epsilon$ & $2B\mu+2\epsilon$ & $2B\mu-2\epsilon$ &
$4B\mu-2\epsilon$ & $6B\mu-6\epsilon$ &  &  &  &  & \\\hline
$\omega$ & $12$ & $2$ & $9$ & $6$ & $6$ & $1$ &  &  &  &  & \\\hline
\end{tabular}
\end{table}

\begin{table}[ptb]
\caption{The eigenvalue $E$ and the degree of degeneracy $\omega$ : $N=7$}%
\label{table5}%
\centering
\begin{tabular}
[c]{ccccccccccccccc}\hline
$E$ & $-7B\mu-7\epsilon$ & $-5B\mu-3\epsilon$ & $-3B\mu+\epsilon$ &
$-3B\mu-3\epsilon$ & $-B\mu+5\epsilon$ & $-B\mu+\epsilon$ &  &  &  &  &  &  &
& \\\hline
$\omega$ & $1$ & $7$ & $14$ & $7$ & $7$ & $21$ &  &  &  &  &  &  &  & \\
$E$ & $-B\mu-3\epsilon$ & $B\mu+5\epsilon$ & $B\mu+\epsilon$ & $B\mu
-3\epsilon$ & $3B\mu+\epsilon$ & $3B\mu-3\epsilon$ &  &  &  &  &  &  &  &
\\\hline
$\omega$ & $7$ & $7$ & $21$ & $7$ & $14$ & $7$ &  &  &  &  &  &  &  & \\
$E$ & $5B\mu-3\epsilon$ & $7B\mu-7\epsilon$ &  &  &  &  &  &  &  &  &  &  &  &
\\\hline
$\omega$ & $7$ & $1$ &  &  &  &  &  &  &  &  &  &  &  & \\\hline
\end{tabular}
\end{table}

\begin{table}[ptb]
\caption{The eigenvalue $E$ and the degree of degeneracy $\omega$ : $N=8$}%
\label{table6}%
\centering
\begin{tabular}
[c]{ccccccccccccccc}\hline
$E$ & $-8B\mu-8\epsilon$ & $-6B\mu-4\epsilon$ & $-4B\mu$ & $-4B\mu-4\epsilon$
& $-2B\mu+4\epsilon$ & $-2B\mu$ &  &  &  &  &  &  &  & \\\hline
$\omega$ & $1$ & $8$ & $20$ & $8$ & $16$ & $32$ &  &  &  &  &  &  &  & \\
$E$ & $-2B\mu-4\epsilon$ & $8\epsilon$ & $-4\epsilon$ & $0$ & $4\epsilon$ &
$2B\mu$ &  &  &  &  &  &  &  & \\\hline
$\omega$ & $8$ & $2$ & $8$ & $36$ & $24$ & $32$ &  &  &  &  &  &  &  & \\
$E$ & $2B\mu+4\epsilon$ & $2B\mu-4\epsilon$ & $4B\mu-4\epsilon$ & $4B\mu$ &
$6B\mu-4\epsilon$ & $8B\mu-8\epsilon$ &  &  &  &  &  &  &  & \\\hline
$\omega$ & $16$ & $8$ & $8$ & $20$ & $8$ & $1$ &  &  &  &  &  &  &  & \\\hline
\end{tabular}
\end{table}

\begin{table}[ptb]
\caption{The eigenvalue $E$ and the degree of degeneracy $\omega$ : $N=9$}%
\label{table7}%
\centering
\begin{tabular}
[c]{ccccccccccccccc}\hline
$E$ & $-9B\mu-9\epsilon$ & $-7B\mu-5\epsilon$ & $-5B\mu-5\epsilon$ &
$-5B\mu-\epsilon$ & $-3B\mu-5\epsilon$ & $-3B\mu-\epsilon$ &  &  &  &  &  &  &
& \\\hline
$\omega$ & $1$ & $9$ & $9$ & $27$ & $9$ & $45$ &  &  &  &  &  &  &  & \\
$E$ & $-3B\mu+3\epsilon$ & $-B\mu-5\epsilon$ & $-B\mu-\epsilon$ &
$-B\mu+3\epsilon$ & $-B\mu+7\epsilon$ & $B\mu-5\epsilon$ &  &  &  &  &  &  &
& \\\hline
$\omega$ & $30$ & $9$ & $54$ & $54$ & $9$ & $9$ &  &  &  &  &  &  &  & \\
$E$ & $B\mu-\epsilon$ & $B\mu+3\epsilon$ & $B\mu+7\epsilon$ & $3B\mu
-5\epsilon$ & $3B\mu-\epsilon$ & $3B\mu+3\epsilon$ &  &  &  &  &  &  &  &
\\\hline
$\omega$ & $54$ & $54$ & $9$ & $9$ & $45$ & $30$ &  &  &  &  &  &  &  & \\
$E$ & $5B\mu-\epsilon$ & $5B\mu-5\epsilon$ & $7B\mu-5\epsilon$ &
$9B\mu-9\epsilon$ &  &  &  &  &  &  &  &  &  & \\\hline
$\omega$ & $27$ & $9$ & $9$ & $1$ &  &  &  &  &  &  &  &  &  & \\\hline
\end{tabular}
\end{table}

\section{The one-dimensional Potts model \label{Potts}}

The Potts model in statistical mechanics is a generalization of the Ising
model, which is a model of interacting spins on a crystalline lattice. The
Hamiltonian of the one-dimensional Potts model is $H=-J\sum_{\left\langle
i,j\right\rangle }\delta\left(  \sigma_{i},\sigma_{j}\right)  $, where
$\sigma_{i}=1,2,3,\ldots,q$, $\delta\left(  \sigma_{i},\sigma_{j}\right)  =1$
if $i=j$, and $\left(  \sigma_{i},\sigma_{j}\right)  =0$ other wise. The
canonical partition function of the one-dimensional Potts model of $N$
particles is \cite{mussardo2010statistical}%
\begin{equation}
Z\left(  \beta,N\right)  =q\left(  q-1+e^{\beta J}\right)  ^{N-1}.
\label{1034}%
\end{equation}
The counting function can be obtained by substituting Eq. (\ref{1034}) into
Eq. (\ref{count_partition}),%
\begin{equation}
N\left(  E\right)  =q\sum_{l=0}^{N-1}\binom{N-1}{l}\left(  q-1\right)
^{N-1-l}\theta\left(  Jl+E\right)  . \label{1035}%
\end{equation}
Substituting Eq. (\ref{1035}) into Eq. (\ref{count}) and solving the equation
give the eigenvalue%
\begin{equation}
E=nJ
\end{equation}
and the degeneracy%
\begin{equation}
\omega=q\binom{N-1}{n}\left(  q-1\right)  ^{N-1-n}-q\binom{N-1}{n-1}\left(
q-1\right)  ^{N-n}.
\end{equation}

\section{Conclusions \label{conclusion}}

In this paper, we suggest a method for calculating the eigenvalue of a
many-body system from the corresponding canonical partition function. The
advantage of the method is that it allows us to merely calculate the
eigenvalue without solving the eigenfunction simultaneously. Recalling that in
many approximate methods, although only needing the eigenvalue, one has to
solve the eigenfunction in the meantime. Solving eigenfunctions, however, is
always a difficult task.

In statistical mechanics, the calculation of thermodynamic quantities only
needs the knowledge of eigenvalues. Only the information of eigenvalue is
embodied in thermodynamic quantities. The method suggested in the present
paper is an approach for extracting the eigenvalue from the thermodynamic
quantity which obtained by statistical mechanical method.

In the present paper, we calculate the eigenvalue from the canonical partition
function. In future works, one can generalizes the method to calculate the
eigenvalue from the other thermodynamic quantities.

Moreover, we improve the result of the relation between the counting function
and the heat kernel given in \cite{dai2009number}.
\appendix

\section{The relation between counting functions and heat kernels
\label{Appendix}}

In Ref. \cite{dai2009number}, we provide a relation between the counting
function
\begin{equation}
N\left(  \lambda\right)  =\sum_{\lambda_{n}<\lambda}1\label{ANlambda}%
\end{equation}
and the global heat kernel
\begin{equation}
K\left(  t\right)  =\sum_{n}e^{-\lambda_{n}t}.\label{AKt}%
\end{equation}
In Ref. \cite{dai2009number}, we only consider the counting function which
counts the number of eigenstates with eigenvalue \textit{smaller} than a given
number $\lambda$, but a special case is ignored: the given number is just a
eigenvalue, i.e., $\lambda=\lambda_{n}$ with $\lambda_{n}$ the $n$-th
eigenvalue. In the following, we provide a complete version of the relation
between $N\left(  \lambda\right)  $ and $K\left(  t\right)  $.

\begin{theorem}%

\begin{numcases}{N\left(  \lambda\right)  =}
\frac{1}{2\pi i}\int_{c-i\infty}^{c+i\infty}K\left(  t\right)  \frac
{e^{\lambda t}}{t}dt,   \text{ when  }\lambda_{n}\neq\lambda \label{A2a}\\
\frac{1}{2\pi i}\int_{c-i\infty}^{c+i\infty}K\left(  t\right)  \frac
{e^{\lambda t}}{t}dt-\frac{1}{2},  \text{ when  }\lambda_{n}=\lambda \label{A2b}%
\end{numcases}

with $c>\lim_{n\rightarrow\infty}\frac{\ln n}{\lambda_{n}}$.
\end{theorem}

\begin{proof}
The function%
\begin{equation}
f\left(  s\right)  =\sum_{n=1}^{\infty}\frac{a_{n}}{\mu_{n}^{s}} \label{A}%
\end{equation}
is a generalization of the Dirichlet series. The function
\begin{equation}
f\left(  s+\omega\right)  =\sum_{n=1}^{\infty}\frac{a_{n}}{\mu_{n}^{s+\omega}}
\label{A8}%
\end{equation}
is uniformly convergent when $\sigma>\beta-c$, where $c=\operatorname{Re}%
\omega$ and $\beta$ is the abscissa of absolute convergence of this Dirichlet
series.\ Performing the integral $\frac{1}{2\pi i}\int_{c-iT}^{c+iT}%
\frac{x^{\omega}}{\omega}d\omega$ on both sides of Eq. (\ref{A8}) gives
\begin{equation}
\frac{1}{2\pi i}\int_{c-iT}^{c+iT}f\left(  s+\omega\right)  \frac{x^{\omega}%
}{\omega}d\omega=\frac{1}{2\pi i}\sum_{n=1}^{\infty}\frac{a_{n}}{\mu_{n}^{s}%
}\int_{c-iT}^{c+iT}\left(  \frac{x}{\mu_{n}}\right)  ^{\omega}\frac{d\omega
}{\omega}. \label{A4}%
\end{equation}
The integral in the right-hand side of Eq. (\ref{A4}), $\frac{1}{2\pi i}%
\int_{c-iT}^{c+iT}\left(  \frac{x}{\mu_{n}}\right)  ^{\omega}\frac{d\omega
}{\omega}$, should be considered in different situations, i.e., $\mu_{n}<x$,
$\mu_{n}=x$, and $\mu_{n}>x$. In the limitation $T\rightarrow\infty$, the
integral reads \cite{tenenbaum2015introduction}
\begin{equation}
\lim_{T\rightarrow\infty}\frac{1}{2\pi i}\int_{c-iT}^{c+iT}\left(  \frac
{x}{\mu_{n}}\right)  ^{\omega}\frac{d\omega}{\omega}=\left\{
\begin{array}
[c]{c}%
\displaystyle1,\text{ }\mu_{n}<x,\\
\displaystyle\frac{1}{2},\text{ }\mu_{n}=x,\\
\displaystyle 0,\text{ }\mu_{n}>x.
\end{array}
\right.  \label{A3}%
\end{equation}
Substituting Eq. (\ref{A3}) into Eq. (\ref{A4}) gives
\begin{equation}
\frac{1}{2\pi i}\int_{c-i\infty}^{c+i\infty}f\left(  s+\omega\right)
\frac{x^{\omega}}{\omega}d\omega=\left\{
\begin{array}
[c]{c}%
\displaystyle\sum_{\mu_{n}<x}^{\infty}\frac{a_{n}}{\mu_{n}^{s}}\text{, when }\mu_{n}\neq
x,\\
\displaystyle\sum_{\mu_{n}<x}^{\infty}\frac{a_{n}}{\mu_{n}^{s}}+\frac{1}{2}\frac{a_{n_{x}}%
}{x^{s}}\text{, when }\mu_{n}=x.
\end{array}
\right.  \label{A88}%
\end{equation}
Setting $s=0$, $a_{n}=1$, $\mu_{n}=e^{\lambda_{n}}$, $x=e^{\lambda}$, and
$\omega=t$ in Eq. (\ref{A88}) gives%
\begin{equation}
\frac{1}{2\pi i}\int_{c-i\infty}^{c+i\infty}f\left(  t\right)  \frac
{e^{\lambda t}}{t}d\omega=\left\{
\begin{array}
[c]{c}%
\displaystyle\sum_{\lambda_{n}<\lambda}^{\infty}1,\text{ when }\lambda_{n}\neq\lambda,\\
\displaystyle\sum_{\lambda_{n}<\lambda}^{\infty}1+\frac{1}{2},\text{ when }\lambda
_{n}=\lambda.
\end{array}
\right.  \label{A99}%
\end{equation}
Comparing the definition of the counting function and the global heat kernel,
Eqs. (\ref{ANlambda}) and (\ref{AKt}), with Eq. (\ref{A99}) proves Eqs.
(\ref{A2a}) and (\ref{A2b}).
\end{proof}

\acknowledgments

We are very indebted to Dr G. Zeitrauman for his encouragement. This work is supported in part by NSF of China under Grant
No. 11575125 and No. 11675119.










\providecommand{\href}[2]{#2}\begingroup\raggedright\endgroup


\end{document}